\documentclass[a4paper,noarxiv,unpublished,onecolumn]{quantumarticle}
\pdfoutput=1

\usepackage[utf8]{inputenc}
\usepackage[english]{babel}
\usepackage[T1]{fontenc}
\usepackage{amsmath}
\DeclareMathOperator*{\argmax}{argmax}
\DeclareMathOperator{\sinc}{sinc}
\DeclareMathOperator{\sincd}{sincd}

\usepackage{hyperref}
\usepackage{tikz}
\usepackage{lipsum}
\usepackage{multirow}

\usepackage[numbers,sort&compress]{natbib}

\usepackage{mathtools}
\usepackage{amssymb}
\usepackage{braket}
\usepackage{amsfonts}

\usepackage{qcircuit}
\usepackage[linesnumbered,ruled,vlined]{algorithm2e}
\usepackage{textcomp}
\usepackage{wasysym}
\usepackage{graphicx}
\usepackage{caption}
\usepackage{subcaption}

\newcommand{\sgn}{\mathop{\mathrm{sgn}}}

\usepackage{amsthm}
\newtheorem{theorem}{Theorem}[section]
\newtheorem{lemma}[theorem]{Lemma}
\newtheorem{corollary}{Corollary}[theorem]
\theoremstyle{definition}

\emergencystretch=\maxdimen
\graphicspath{ {./images/} }

\begin{document}

\title{Properties of The Discrete Sinc Quantum State and Applications to Measurement Interpolation}

\author{Charlee Stefanski}
\orcid{0000-0001-9856-5955}
\affiliation{Wells Fargo}
\affiliation{UC Berkeley}
\author{Vanio Markov}
\affiliation{Wells Fargo}
\author{Constantin Gonciulea}
\orcid{0000-0001-5870-4586}
\affiliation{Wells Fargo}

\maketitle
\begin{abstract}

\end{abstract}

\maketitle

\begin{abstract}
    Extracting the outcome of a quantum computation is a difficult task.
    In many cases, the quantum phase estimation algorithm is used to digitally encode a value in a quantum register
    whose amplitudes' magnitudes reflect the discrete $\sinc$ function.
    In the standard implementation the value is approximated by the most frequent outcome, however, using the
    frequencies of other outcomes allows for increased precision without using additional qubits.
    One existing approach is to use Maximum Likelihood Estimation, which uses the frequencies of all measurement
    outcomes.
    We provide and analyze several alternative estimators, the best of which rely on only the two most frequent
    measurement outcomes.
    The Ratio-Based Estimator uses a closed form expression for the decimal part of the encoded value using the ratio of
    the two most frequent outcomes.
    The Coin Approximation Estimator relies on the fact that the decimal part of the encoded value is very well
    approximated by the parameter of the Bernoulli process represented by the magnitudes of the largest two amplitudes.
    We also provide additional properties of the discrete $\sinc$ state that could be used to design other estimators.
\end{abstract}

\section{\label{sec:intro}Introduction}

Quantum phase estimation is a fundamental method in quantum computing, used as a building block in many other quantum
algorithms, such as Shor's and quantum amplitude estimation.
Its core underlying procedure first creates an analog representation of a periodic signal into quantum state, and
then digitally encodes the period of this signal into the state of a quantum register that can  be efficiently measured.

We refer to this quantum state with remarkable properties as ``discrete sinc'', the ``period encoding state'', the
``phase estimation state'', or the ``interpolation state'', because after a phase correction, its amplitudes match
the interpolation coefficients in the classical interpolation theorem ~\cite{uiocourse, interpolation}.
We call the resulting outcome probability distribution the ``discrete sinc squared'' distribution or ``the Fej\'er
distribution'' because its probabilities match the coefficients in Fej\'er kernels~\cite{Hoffman2007, Janson2010}.
We used this state and the underlying procedure in the quantum phase estimation algorithm to encode discrete
functions, or dictionaries, into quantum state, and interpolate non-integer values ~\cite{interpolation}.

\begin{figure}[ht]
    \centering
    \begin{minipage}{.3\textwidth}
        \centering
        \includegraphics[width=.99\linewidth]{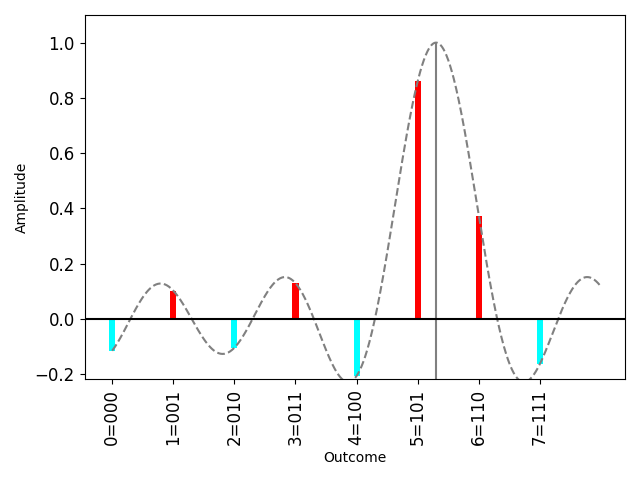}
    \end{minipage}
    \begin{minipage}{.3\textwidth}
        \centering
        \includegraphics[width=.99\linewidth]{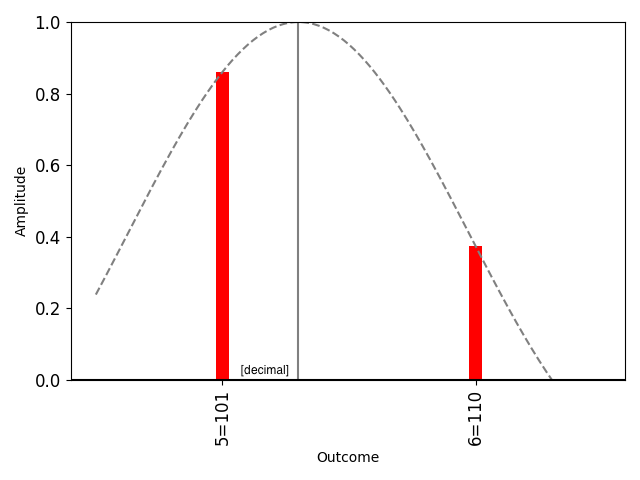}
    \end{minipage}
    \captionof{figure}{Left: Histogram of the amplitudes of a 3-qubit quantum state with the encoded value 5.3
    compared to the discrete sinc function.
    Right: The two largest amplitudes of the state, which can be used to estimate the decimal part of the encoded
    value.}
    \label{fig:encoding_example}
\end{figure}

While the canonical phase/amplitude estimation algorithms use digital encoding of a value to estimate it, efforts
have been made to improve the accuracy of those estimates by interpolating ``in-between'' the discrete values. In
particular, Maximum Likelihood Estimation has been used to post-process measurement results and improve the
estimation precision without additional qubits~\cite{IQAE}.

In this paper we present additional properties of the discrete sinc quantum state and Fej\'er distribution
and provide alternative estimation methods.
We provide closed-form estimators for the expressions for the encoded value, that use consecutive
pairs of amplitudes as inputs.
The pair with the highest magnitudes is the most useful, and we show that it can be used to represent a Bernoulli
process, i.e. a (biased) coin, whose parameter (bias) is an estimate for the decimal part of the encoded value. We
also show that the equation that needs to be solved in order to find the Maximum Likelihood Estimate is a form of
interpolation in the context of the Fej\'er distribution.

\section{\label{sec:prelim}Preliminaries}

The $\sinc$ function is common in digital signal processing~\cite{sinc}.
It can be defined on a set of real numbers as

\begin{equation*}
    \begin{split}
        \sinc(t) & =
        \begin{cases}
            1, & \text{if } t = 0 \\
            \frac{\sin(t)}{t}, & \text{otherwise} \\
        \end{cases} \\
        & = \prod_{j = 1}^{\infty} \cos \left(\frac{t}{2^j} \right)
    \end{split}
\end{equation*}

For a positive integer $n$ and a real number $t$ the function $\sincd_n$ is defined as

\begin{equation*}
    \sincd_{n}(t) = \frac{\sinc(t)}{\sinc\left(\frac{t}{2^n}\right)} = \prod_{j = 1}^{n} \cos \left(\frac{t}{2^j}
    \right).
\end{equation*}

For a real number $t$, the normalized versions of the $\sinc$ and $\sincd$ functions are defined as $\sinc_{\pi}(t) =
\sinc(\pi t)$ and $\sincd_{\pi, n}(t) = \sincd_{n}(\pi t)$, respectively.

Given a positive integer $n$ and a real number $t \in [0, N)$, where $N=2^n$, consider the quantum state

\begin{equation}
    \ket{\phi_{n,t}} = \sum_{k = 0}^{N-1} e^{i \pi \frac{N-1}{N} (t - k)} c_{N,t}(k) \ket{k}_n
    \label{eqn:fejer_state}
\end{equation}

where

\begin{equation*}
    \begin{split}
        c_{N,t}(k) & =
        \begin{cases}
            1, & \text{if } t = k \\
            \frac{1}{N} \frac{\sin((t - k)\pi)}{\sin((t-k)\frac{\pi}{N})}, & \text{otherwise} \\
        \end{cases} \\
        & = \sincd_{\pi, n}(t-k) \\
        & = \prod_{j = 1}^{n} \cos \left((t - k) \frac{\pi}{2^j} \right).
    \end{split}
\end{equation*}

Note that the amplitudes in the state add up to 1:
\begin{equation}
    \sum_{k = 0}^{N-1} e^{i \pi \frac{N-1}{N} (t - k)} c_{N,t}(k)  = 1.
    \label{eqn:identity_1}
\end{equation}

This state encodes the result of the phase estimation algorithm, also used in the amplitude estimation algorithm,
and to encode values and functions ~\cite{interpolation, inner_product}.

If $t \in [0, N)$ is not an integer, the probability mass function of the measurement distribution for the
$\ket{\phi_{n,t}}$ quantum state is

\begin{equation}
    \label{eqn:fejer_pdf}
    p_{N, t}(k) = \frac{1}{N^2} \frac{\sin^2((t - k)\pi)}{\sin^2((t-k)\frac{\pi}{N})}
\end{equation}

for $0 \le k < N$.

The values of $p_{N, t}$ match those of the normalized Fej\'er kernel~\cite{Hoffman2007, Janson2010}.

\begin{lemma}[MLE Property]
    \label{mle_lemma}
    With the notations above, the following identity holds for a non-integer $t \in [0, N)$:

    \begin{equation}
        \label{eqn:mle_identity}
        \frac{1}{N}\sum_{k=0}^{N-1} p_{N, t}(k)\cot((t-k)\frac{\pi}{N}) = \cot(t\pi)
    \end{equation}
\end{lemma}

This property allows for the estimation of the non-integer parameter $t \in [0, N)$ of a given quantum state
$\ket{\phi_{n,t}}$ by repeated measurement ~\cite{IQAE}.
The estimation as a real number is more precise than the one obtained by just using the integer outcomes of a
measurement, as in the standard phase estimation algorithm.

Equivalent forms of this equation are:

\begin{equation*}
    \frac{1}{N}\sum_{k=0}^{N-1} \left(1 - p_{N, t}(k)\right)\cot((t-k)\frac{\pi}{N}) = 0,
\end{equation*}

and

\begin{equation*}
    \sum_{k=0}^{N-1} p_{N, t}(k) (-1)^k \cos((t-k)\frac{\pi}{N}) c_{N, t}(k) = \cos(t\pi).
\end{equation*}

Estimating the parameter $t$ of the period encoding state $\ket{\phi_{n, t}}$ and its corresponding probability
distribution $p_{N, t}$ from the function $q$ is a statistical inference task. We are looking for the estimate
$\hat{t} \in [0, N)$ such that the distribution $p_{N, \hat{t}}$ is the best fit for the function $q$ obtained
through measurement.

\paragraph{Kullback–Leibler divergence or maximum likelihood estimation.}

The relative entropy, or the Kullback–Leibler divergence, from $q$ to $p_{N, t}$ is

\begin{equation*}
    \begin{split}
        D_{KL}(q \parallel p_{N, t}) & = \sum_{k=0}^{N-1} q(k) \log \left(\frac{q(k)}{p_{N, t}(k)}\right) \\
        & = \sum_{k=0}^{N-1} q(k) \log (q(k)) - \log \mathcal{L}(t \lvert q)
    \end{split}
\end{equation*}

where

\begin{equation*}
    \begin{split}
        \mathcal{L}(t \lvert q) & = \prod_{k=0}^{N-1} p_{N, t}(k)^{q(k)} \\
    \end{split}
\end{equation*}

is the likelihood of the parameter $t$ given the measurement reflected in the function $q$.

Minimizing the Kullback–Leibler divergence is the same as maximizing the log-likelihood function, which is
the essence of the Maximum Likelihood Estimation method:

\begin{equation*}
    \hat{t}_{\text{MLE}} = \argmax_{t \in [0, N)} \mathcal{L}(t \lvert q)
\end{equation*}

Setting the derivative of the log-likelihood function to zero gives the equation for Maximum Likelihood Estimate:

\begin{equation}
    \label{eqn:mle_equation}
    \begin{split}
        0 & = \frac{\partial}{\partial t} \log\mathcal{L}(t \lvert q)(\hat{t}_{\text{MLE}}) \\
        0 & = \frac{\partial}{\partial t} \sum_{k=0}^{N-1} q(k) \log (p_{N, t}(k))(\hat{t}_{\text{MLE}}) \\
        \cot(\hat{t}_{\text{MLE}}\pi) & = \frac{1}{N}\sum_{k=0}^{N-1} q(k)\cot((\hat{t}_{\text{MLE}}-k)\frac{\pi}{N})
    \end{split}
\end{equation}

Note that if $q = p_{N, t}$ (an ideal measurement) then $\hat{t}$ satisfies the equation as reflected in the identity in
the preliminaries (Eq.~\ref{eqn:mle_equation}).

This method has the benefit of well-understood theory, including confidence intervals.
However, the likelihood function built from measurements on quantum devices in the NISQ era may diverge significantly
from the true one.
Figure~\ref{fig:likelihoods} plots the likelihood function for various $t$ values using experiments run on real
quantum devices.
These experiments are discussed in detail in Section~\ref{sec:numerical_exp}.

\begin{figure*}[ht]
    \centering
    \begin{minipage}{.3\textwidth}
        \centering
        \includegraphics[width=.99\linewidth,trim= 40pt 40pt 40pt 40pt]{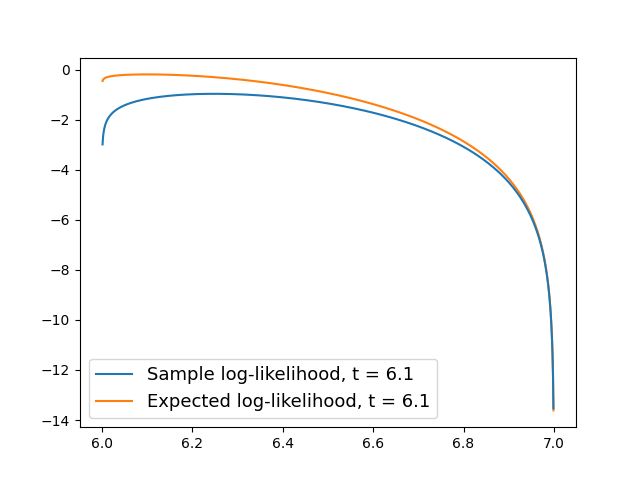}
    \end{minipage}
    \begin{minipage}{.3\textwidth}
        \centering
        \includegraphics[width=.99\linewidth, trim= 40pt 40pt 40pt 40pt]{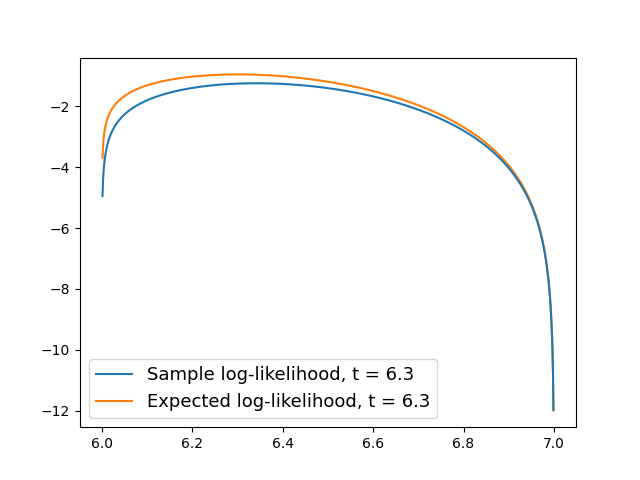}
    \end{minipage}
    \begin{minipage}{.3\textwidth}
        \centering
        \includegraphics[width=.99\linewidth, trim= 40pt 40pt 40pt 40pt]{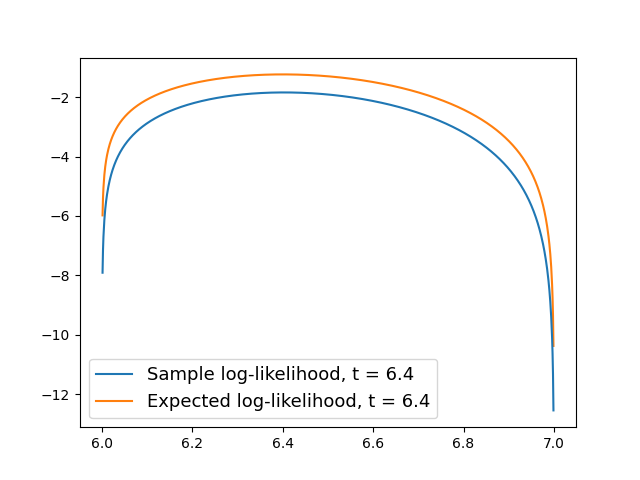}
    \end{minipage}

    \bigskip
    \centering
    \begin{minipage}{.3\textwidth}
        \centering
        \includegraphics[width=.99\linewidth, trim= 40pt 40pt 40pt 40pt]{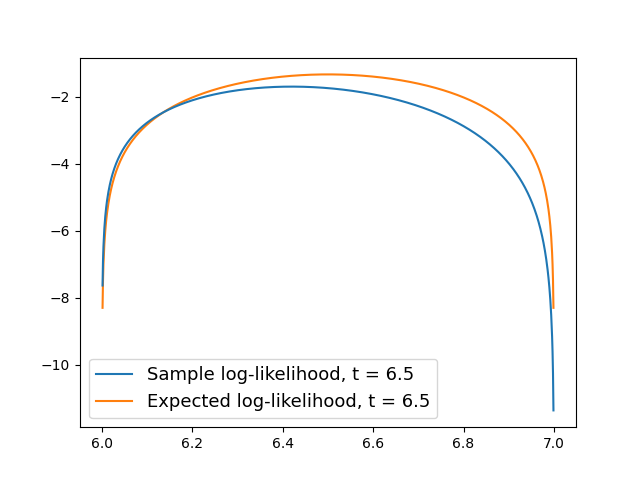}
    \end{minipage}
    \begin{minipage}{.3\textwidth}
        \centering
        \includegraphics[width=.99\linewidth, trim= 40pt 40pt 40pt 40pt]{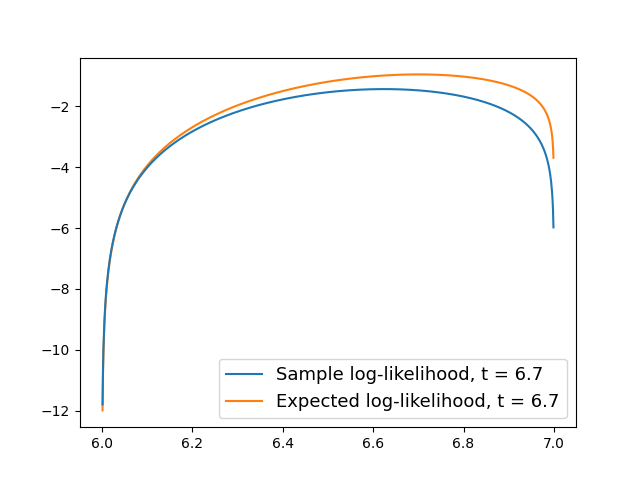}
    \end{minipage}
    \begin{minipage}{.3\textwidth}
        \centering
        \includegraphics[width=.99\linewidth, trim= 40pt 40pt 40pt 40pt]{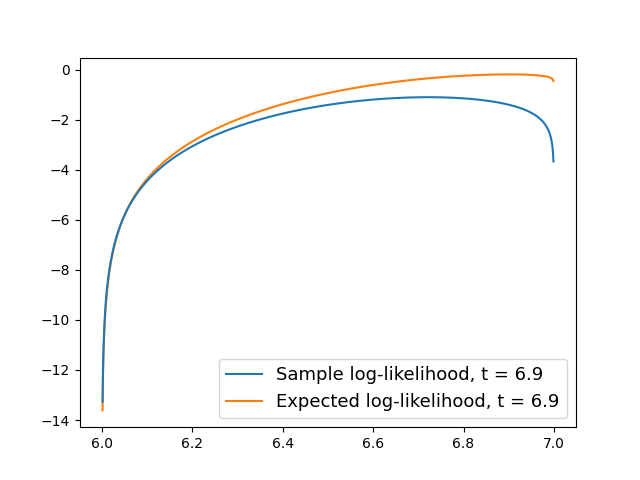}
    \end{minipage}
    \captionof{figure}{The theoretical log-likelihood function compared to the log-likelihood using measurement
    probabilities from experiments on \textit{ibm\_perth} $\log\mathcal{L}(t\lvert q)$ for several values $t$.}
    \label{fig:likelihoods}
\end{figure*}

\section{\label{sec:methods}Discrete Sinc Quantum State and Fej\'er Distribution Properties}

In this section we provide properties of the discrete $\sinc$ quantum state that are essential to the design of estimators
in the next section.

\begin{lemma}
    \label{lemma1}
    For a non-integer value $t \in (0, N)$ and an integer $0 \le k < N-1$ we have

    \begin{equation*}
            c_{N, t}(k) = (-1)^{\lfloor t \rfloor - k}\sgn(\lfloor t \rfloor - k)\sqrt{p_{N, t}(k)}
    \end{equation*}

    where $\sgn(x) = -1$ if $x < 0$ and $\sgn(x) = 1$ if $x \ge 0$.
\end{lemma}

\begin{proof}

    \begin{equation*}
        \begin{split}
            c_{N, t}(k) & = \frac{1}{N} \frac{\sin((t-k)\pi)}{\sin((t-k)\frac{\pi}{N})} \\
            & = \frac{1}{N} \frac{(-1)^{\lfloor t \rfloor - k}|\sin(t\pi)|}{\sgn(\lfloor t \rfloor - k)|\sin((t-k)
                \frac{\pi}{N})|} \\
            & = (-1)^{\lfloor t \rfloor - k}\sgn(\lfloor t \rfloor - k)\sqrt{p_{N, t}(k)}
        \end{split}
    \end{equation*}
\end{proof}

\begin{lemma}[Ratio-Based Estimation]
    \label{re_lemma}
    Given a non-integer value $t \in (0, N)$ we have

    \begin{equation}
        \begin{split}
            \tan((t-k)\frac{\pi}{N}) &= \frac{\sin\frac{\pi}{N}}{\cos\frac{\pi}{N} + \sqrt{\frac{p_{N, t}(k)}{p_{N,
            t}(k+1)}}} ,
        \end{split}
        \label{eqn:re_formula}
    \end{equation}

    for an integer $0 \le k < N-1$, and

    \begin{equation}
        \begin{split}
            \tan((t-N + 1)\frac{\pi}{N}) &= \frac{\sin\frac{\pi}{N}}{\cos\frac{\pi}{N} - \sqrt{\frac{p_{N, t}(N-1)
            }{p_{N, t}(0)}}}.
        \end{split}
        \label{eqn:re_formula1}
    \end{equation}

\end{lemma}

\begin{proof}
    The following proves Eq.~\ref{eqn:re_formula}:

    \begin{equation*}
        \begin{split}
            \sqrt{\frac{p_{N, t}(k)}{p_{N, t}(k+1)}} & =\frac{c_{N, t}(k)}{c_{N, t}(k+1)} \\
            & = - \frac{\sin((t-k-1)\frac{\pi}{N})}{\sin((t-k)\frac{\pi}{N})} \\
            & = - \frac{\sin((t-k)\frac{\pi}{N})\cos\frac{\pi}{N} - \cos((t-k)\frac{\pi}{N})\sin\frac{\pi}{N}}{\sin((t-k)\frac{\pi}{N})} \\
            & = - \cos\frac{\pi}{N} + \frac{1}{\tan((t-k)\frac{\pi}{N})}\sin\frac{\pi}{N}
        \end{split}
    \end{equation*}

    Therefore,
    \begin{equation*}
        \begin{split}
            \tan((t-k)\frac{\pi}{N}) &= \frac{\sin\frac{\pi}{N}}{\cos\frac{\pi}{N} + \sqrt{\frac{p_{N, t}(k)}{p_{N,
            t}(k+1)}}}.
        \end{split}
    \end{equation*}

\end{proof}

\begin{corollary}
    \label{ratio_corollary}
    For $t \in (k, k+1)$ and an integer $0 \le k < N-1$ we have
    \begin{equation*}
        t = k + \frac{N}{\pi} \arctan\left(\frac{\sin\frac{\pi}{N}}{\cos\frac{\pi}{N} + \sqrt{\frac{p_{N, t}(k)}{p_{N, t}(k+1)}}}\right).
    \end{equation*}
\end{corollary}

If $t \in (N-1, N)$ an adjustment needs to be made to the formula:

\begin{equation*}
    t = N-1 + \frac{N}{\pi} \arctan\left(\frac{\sin\frac{\pi}{N}}{\cos\frac{\pi}{N} - \sqrt{\frac{p_{N, t}(N-1)}{p_{N, t}(0)}}}\right).
\end{equation*}

These are also an analytical solutions for Eq.~\ref{eqn:mle_equation}.

\begin{lemma}[Coin Approximation]
    \label{coin_lemma}
    Given a non-integer value $t \in (k, k+1)$ and an integer $0 \le k < N-1$, if N is sufficiently large
    the decimal part of $t$ can be approximated by

\begin{equation*}
    \begin{split}
        t-k & \approx \frac{1}{1 + \sqrt{\frac{p_{N, t}(k)}{p_{N, t}(k+1)}}}\\
        & = \frac{\sqrt{p_{N, t}(k+1)}}{\sqrt{p_{N, t}(k)} + \sqrt{p_{N, t}(k+1)}}.
    \end{split}
\end{equation*}
\end{lemma}

\begin{proof}
    Given the fact that when $N$ is sufficiently large, $\frac{1}{N}$ is close to $0$, we can approximate $\sin
    \left(\frac{\pi}{N}\right)$ by $\frac{\pi}{N}$, $\tan\left((t-k)\frac{\pi}{N}\right)$ by $(t-k)\frac{\pi}{N}$, and
    $\cos\left(\frac{\pi}{N}\right)$ by $1$ in Eq.~\ref{eqn:re_formula} and Eq.~\ref{eqn:re_formula1}.
\end{proof}

\begin{corollary} The decimal part of $t$ can be approximated as the bias of a coin that lands heads-up $\lfloor
M\sqrt{p_{N, t}(k+1)} \rfloor$ times and tails-up $\lfloor M\sqrt{p_{N, t}(k)} \rfloor$ times, where the integer
factor M is chosen based on the desired precision.

\end{corollary}

\begin{lemma}[Interpolation Formula]
    \label{lemma_interpolation}
    Combining the classical interpolation theorem~\cite{uiocourse} and Lemma~\ref{lemma1}, for a well-behaved
    (periodic, band limited as in~\cite{interpolation}) function $f: \{0, \mathellipsis, N-1\} \rightarrow \mathbb{R}$:

    \begin{equation*}
        f(t) = \sum_{k = 0}^{N-1} f\left(\frac{k}{N}\right) c_{N, t}(k) = \sum_{k = 0}^{N-1} f\left
        (\frac{k}{N}\right) (-1)^{\lfloor t \rfloor - k}\sgn(\lfloor t \rfloor - k)\sqrt{p_{N, t}(k)}
    \end{equation*}

\end{lemma}

\section{\label{sec:applications}Applications to Quantum Measurement Interpolation}

Assume we have an $n$-qubit quantum register whose state is the result of encoding a value, following the quantum
phase estimation procedure.
The value could represent the phase of a unitary operator's eigenvalue, as in the original context of quantum phase
estimation, the probability of marked states, as in the quantum amplitude estimation context, encoding function values, etc.

With the notation $N=2^n$, we interpret the encoded value as a real number $t \in [0, N)$.
In some contexts we  may be interested in the value $\frac{t}{N} \in [0, 1)$.
The state of the register will reflect the quantum state in Eq.~\ref{eqn:fejer_state}.

Repeated measurements of the state of the register create a sample from the probability distribution $p_{N, t}$ as in
Eq.~\ref{eqn:fejer_pdf}.
We denote by $q$ the function that maps the outcome $k$, for $0 \le k < N$, corresponding to the computational state
$\ket{k}_n$ to the proportion (normalized count) of measurements of the outcome $k$.

\paragraph{Ratio-Based Estimation}

As discussed in Lemma~\ref{re_lemma}, using the formulas in Corollary~\ref{ratio_corollary}, the ratio of amplitudes
can be used to get an estimate for the value $t$:

\begin{equation}
    \hat{t}_\text{RBE} = \lfloor t \rfloor + \frac{N}{\pi} \arctan\left(\frac{\sin\frac{\pi}{N}}{\cos\frac{\pi}{N} +
    \sqrt{q(\lfloor t \rfloor)/q(\lceil t \rceil)}}\right)
    \label{eqn:ratio_est}
\end{equation}

or

\begin{equation}
    \hat{t}_\text{RBE} = \lceil t \rceil - \frac{N}{\pi} \arctan\left(\frac{\sin\frac{\pi}{N}}{\cos\frac{\pi}{N} +
    \sqrt{q(\lceil t \rceil)/q(\lfloor t \rfloor)}}\right),
    \label{eqn:ratio_est_ceil}
\end{equation}

where we can infer the ceiling and floor of $t$ from the measurements (top two largest values of $q$). We abbreviate
this method by "RBE".

\paragraph{Coin Approximation Estimation}

As discussed in Lemma~\ref{coin_lemma}, the magnitudes of the amplitudes of the floor and ceiling of the value $t$ (i
.e. the amplitudes with the largest magnitudes) can be used as likelihoods for the sides of a coin whose bias is an
estimate for $t - \lfloor t \rfloor$, the decimal part of $t$.

We can use one of many available methods to estimate the bias of this Bernoulli process.
This approach has the benefit of well-understood theory and known confidence/credible intervals.

In our experiments we have used the Bayesian approach that relies on the fact that the Beta distribution is the
conjugate prior of the Bernoulli distribution.
The posterior distribution is the Beta distribution with the square roots of the two largest sampling counts as
parameters.

\paragraph{Interpolation-Based Estimation}

Using the interpolation formula in Lemma~\ref{lemma_interpolation} with a well-behaved function (e.g $x \mapsto \cos
(x\pi)$) we get

\begin{equation*}
    \cos(\hat{t}\pi) \approx \sum_{k = 0}^{N-1} \cos\left(\frac{k}{N}\pi\right) (-1)^{\lfloor t \rfloor - k}\sgn
    (\lfloor t \rfloor - k)\sqrt{q(k)}.
\end{equation*}

Then we can solve for $\hat{t}$.

\section{\label{sec:confidence_interval}Confidence Intervals for Ratio-Based and Coin Approximation Estimators}

For an integer $n > 1$ and a real number $t \in (k, k+1)$, where $N=2^n$ and $0 \le k < N$, consider the state
$\ket{\phi_{n, t}}$ defined in Eq.~\ref{eqn:fejer_state}.
For a positive integer $L$, and a sequence of $L$ measurements of this state, denote by $r$ the ratio of the
normalized measurement frequencies of the states $\ket{k}$ and $\ket{k+1}$, and by $s$ the sum of these frequencies.
Then, according to Eq.~\ref{eqn:ratio_est}, the decimal part of $t$ depends only on $N$ and $r$, and not on $k$,
and is defined by:

\begin{equation}
    \label{eqn:d_definition}
    D_N(r) = \frac{N}{\pi} \arctan\left(\frac{\sin\frac{\pi}{N}}{\cos\frac{\pi}{N} + \sqrt{r}}\right).
\end{equation}

Since the ratio estimator is asymptotically normally distributed ~\cite{meanandvariance2018}, we can use the delta
method ~\cite{statisticalinference} to derive a $100(1-\alpha)$\% confidence interval for the estimator of the
decimal part of $t$:

\begin{equation*}
    \left[ D_N(r) - z_{\alpha/2} \frac{\sigma D_N^{\prime}(r)}{L},
    D_N(r) + z_{\alpha/2} \frac{\sigma D_N^{\prime}(r)}{L} \right]
\end{equation*}

where $\sigma = (1+r)\sqrt{\frac{r}{s}}$, and  $z_{\alpha/2}$ denotes the normal critical value corresponding to the
significance level $\alpha \in (0, 1)$.

Credible intervals for the Coin Approximation Estimator can be computed using the percent point function of the Beta
distribution.
Figure~\ref{fig:delta_beta} compares the radius of the intervals using these two methods.

\begin{center}
    \centering
    \begin{minipage}{.43\textwidth}
        \centering
        \includegraphics[width=.99\linewidth]{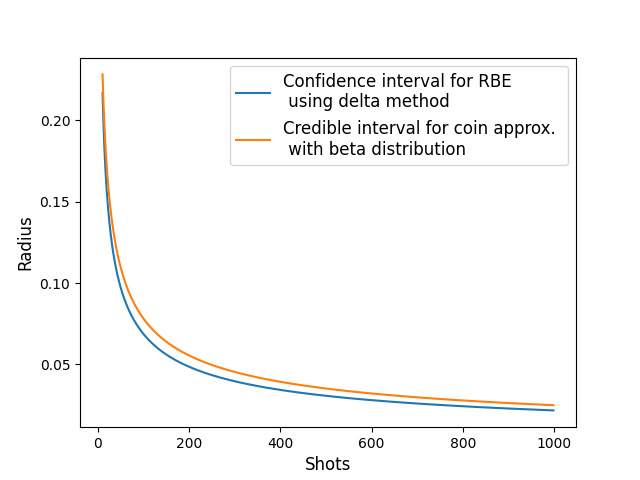}
    \end{minipage}
    \captionof{figure}{The radius of the RBE confidence interval compared to the Coin Approx. credible interval
    using the theoretical measurement frequencies of the highest two measurements given $n = 3, N = 8$ and $t = 4
    .2$, for a given number of shots.}
    \label{fig:delta_beta}
\end{center}

\section{\label{sec:numerical_exp}Numerical Experiments}

In this section we perform experiments for an empirical analysis of the methods described in the previous section on
both quantum simulators and quantum computers.
We use IBM Quantum services to perform the experiments ~\cite{IBMQServices}.
Details about the systems and configurations can be found in Appendix~\ref{sec:hardware}.

Using the value encoding algorithm described in ~\cite{interpolation} to prepare the state $\ket{\phi_{n, t}}$ as
defined in Eq.~\ref{eqn:fejer_state} using $n$-qubits and a given value $t \in [0, 2^n)$.
Each circuit is run with $20,000$ shots.
The memory parameter in IBM Quantum services~\cite{Qiskit} allows for the measurement at each shot to be retrieved.
We refer to the measuremnt at a given shot as a sample.

Given a round of samples, we estimate the parameter $t$ using the MLE method discussed in Lemma~\ref{mle_lemma} (Eq
.~\ref{eqn:mle_equation}), as well as the Ratio-Based Estimation (RBE) and Coin Approximation estimation methods
introduced Section~\ref{sec:applications}. Interpolation-Based Estimation is not included because it did not perform
well in experiments.

\subsection{\label{subsec:sim}Quantum Simulator Experiments}

The estimations from experiments performed on a quantum simulator backend were both highly accurate and highly
precise.
Figure~\ref{fig:sim_estimations_t6} visualizes 20 estimations for values of $t$ at increments of 0.1 in between 6.1
and 6.9 using three different estimation methods: RBE, Coin Approximation and MLE. The average error of the 20
estimates is low for all three methods, and it is difficult to determine if any method gives better estimates.

\begin{figure*}[ht]
    \centering
    \begin{minipage}{.4\textwidth}
        \centering
        \includegraphics[width=1.0\linewidth, trim = 20pt 20pt 20pt 20pt]{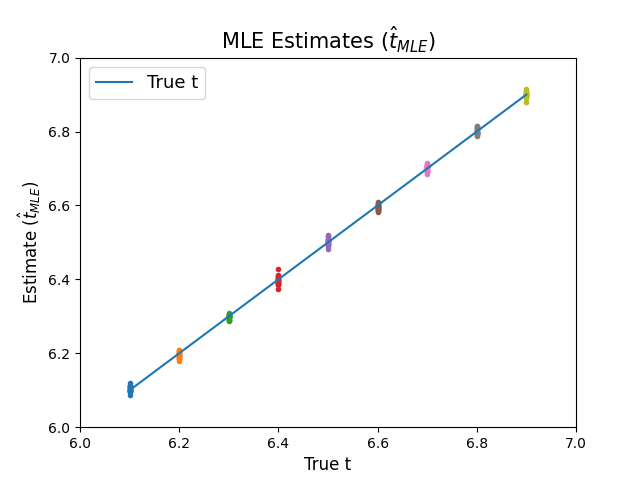}
    \end{minipage}
    \begin{minipage}{.4\textwidth}
        \centering
        \includegraphics[width=1.0\linewidth, trim = 20pt 20pt 20pt 20pt]{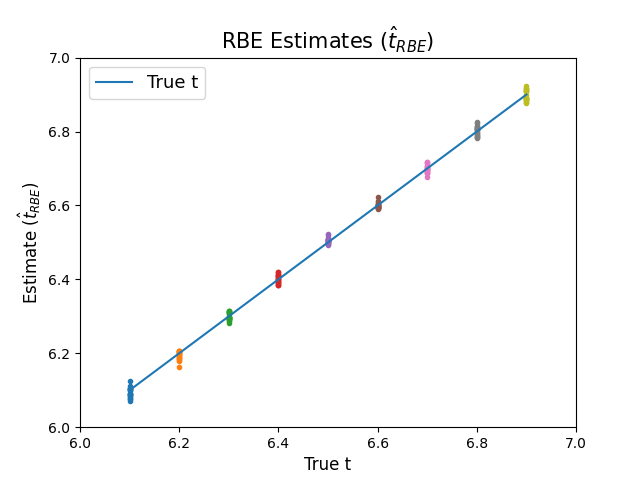}
    \end{minipage}

    \bigskip
    \begin{minipage}{.4\textwidth}
        \centering
        \includegraphics[width=1.0\linewidth, trim = 20pt 20pt 20pt 20pt]{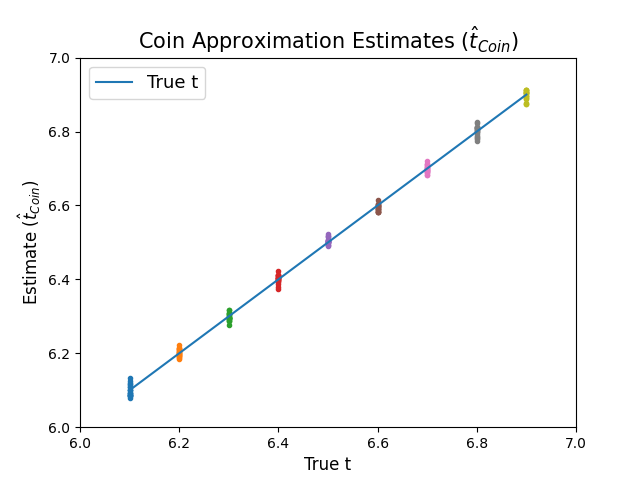}
    \end{minipage}
    \begin{minipage}{.4\textwidth}
        \centering
        \includegraphics[width=1.0\linewidth, trim = 20pt 20pt 20pt 20pt]{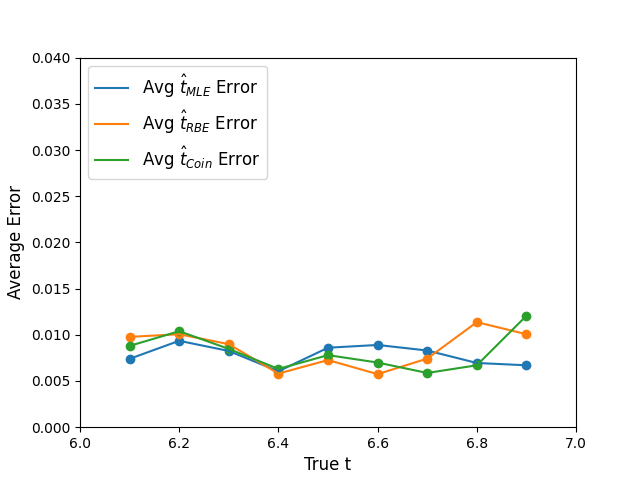}
    \end{minipage}

    \captionof{figure}{Top Left: MLE estimates for $t$ computed using 1,000 samples from quantum simulator
    experiments with $n = 3$ qubits and values of $t$ at increments of 0.1 in between 6.1
    and 6.9. Top Right: RBE estimates for $t$ computed using 1,000 samples from quantum simulator experiments with
        $n = 3$ qubits and values of $t$ at increments of 0.1 in between 6.1 and 6.9. Bottom Left: Coin
        Approximation estimates for $t$ computed using 1,000 samples from quantum simulator experiments with $n = 3$
        qubits and values of $t$ at increments of 0.1 in between 6.1 and 6.9. Bottom Right: The average error of
        estimates computed using samples from experiments on a quantum simulator.}
    \label{fig:sim_estimations_t6}
\end{figure*}

\subsection{\label{subsec:real_experiments}Quantum Hardware Experiments}

Figures~\ref{fig:zoom_mle_re} and~\ref{fig:zoom_mle_coin} visualize the measurement frequencies from experiments
using 3 qubits on \textit{ibm\_perth} and estimates for various values of $t$. Each estimate is computed using a set
of $1,000$ samples, and we repeat the process for 20 rounds.
Figure~\ref{fig:zoom_mle_re} compares estimates computed using MLE and RBE methods and
Figure~\ref{fig:zoom_mle_coin} compares estimates computed using MLE and Coin Approximation Estimation.

\begin{center}
    \begin{tabular}{ccc}
        \begin{minipage}{.32\textwidth}
            \centering
            \includegraphics[width=1.0\linewidth, trim = 10pt 10pt 10pt 10pt]{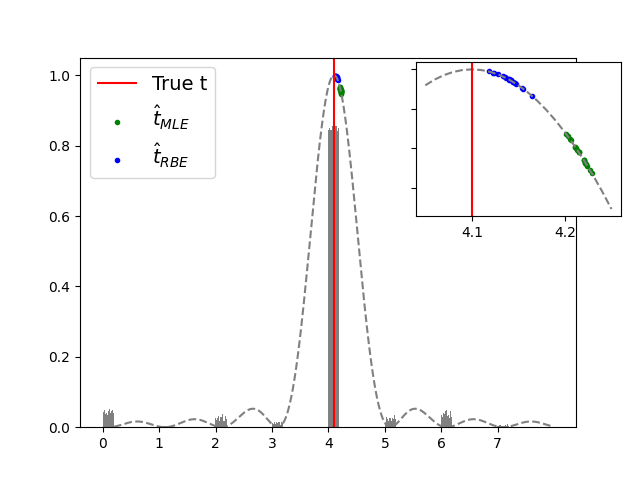}
        \end{minipage}
        \begin{minipage}{.32\textwidth}
            \centering
            \includegraphics[width=1.0\linewidth, trim = 10pt 10pt 10pt 10pt]{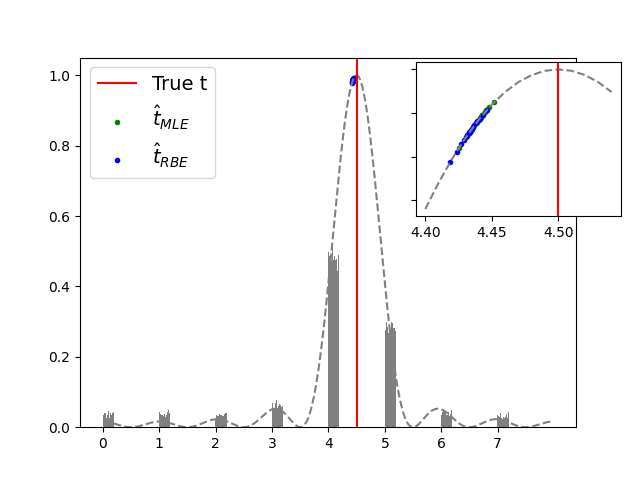}
        \end{minipage}
        \begin{minipage}{.32\textwidth}
            \centering
            \includegraphics[width=1.0\linewidth, trim = 10pt 10pt 10pt 10pt]{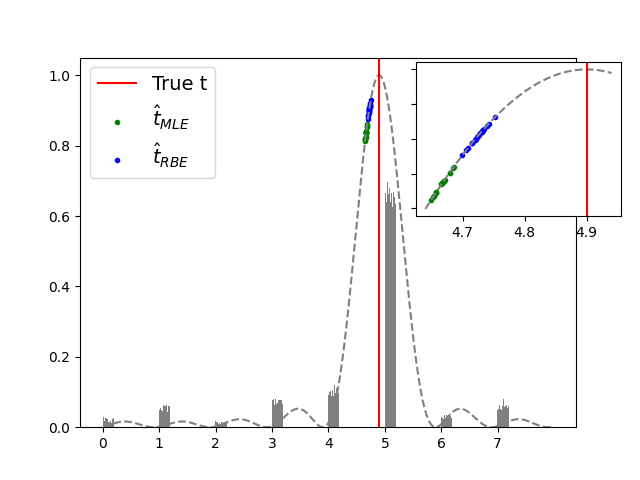}
        \end{minipage}
    \end{tabular}
    \captionof{figure}{The histograms of normalized measurement probabilities for each set of 1,000 samples from
    experiments with $n = 3$ qubits and $t = 4.1$, $t = 4.5$ and $t = 4.9$ on \textit{ibm\_perth}, compared to the
    expected probability distribution defined in Eq.~\ref{eqn:fejer_pdf}. For each set of 1,000 samples, we
    compute estimates using MLE and RBE.}
    \label{fig:zoom_mle_re}
\end{center}

\begin{center}
    \begin{tabular}{ccc}
        \begin{minipage}{.32\textwidth}
            \centering
            \includegraphics[width=1.0\linewidth, trim = 10pt 10pt 10pt 10pt]{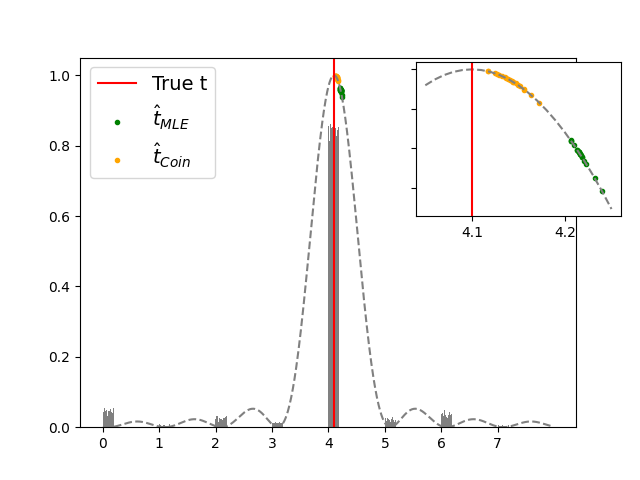}
        \end{minipage}
        \begin{minipage}{.32\textwidth}
            \centering
            \includegraphics[width=1.0\linewidth, trim = 10pt 10pt 10pt 10pt]{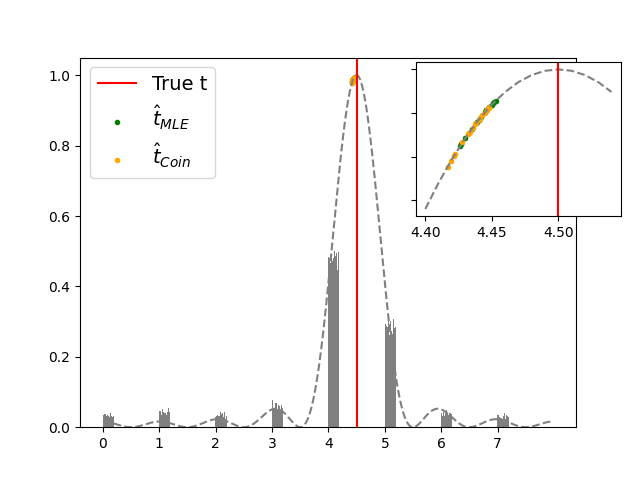}
        \end{minipage}
        \begin{minipage}{.32\textwidth}
            \centering
            \includegraphics[width=1.0\linewidth, trim = 10pt 10pt 10pt 10pt]{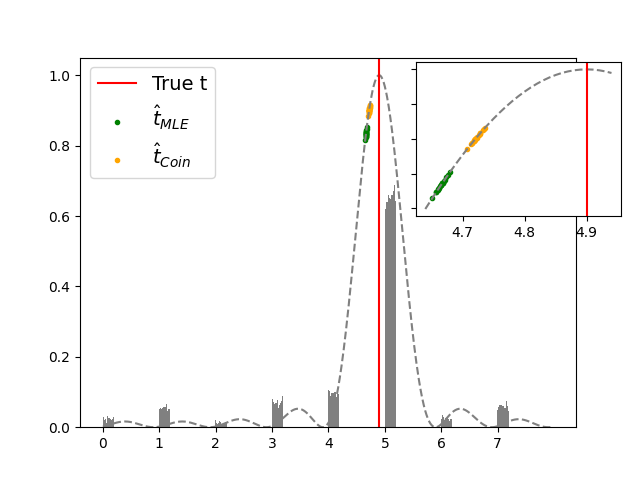}
        \end{minipage}
    \end{tabular}
    \captionof{figure}{The histograms of normalized measurement probabilities for each set of 1,000 samples from
    experiments with $n = 3$ qubits and $t = 4.1$, $t = 4.5$ and $t = 4.9$ on \textit{ibm\_perth}, compared to the
    expected probability distribution defined in Eq.~\ref{eqn:fejer_pdf}. For each set of 1,000 samples, we
    compute estimates using MLE and Coin Approximation Estimation.}
    \label{fig:zoom_mle_coin}
\end{center}

As mentioned in Section~\ref{sec:prelim}, results from experiments on real quantum hardware in the NISQ era show much
more divergence from the true likelihood function than results from quantum simulation.
Figure~\ref{fig:perth_estimations_t6} visualizes results from 20 rounds with $1,000$ samples each for values of $t$ at
increments of 0.1 in between 6.1 and 6.9 using a 3 qubits on \textit{ibm\_perth}.

\begin{figure*}[ht]
    \centering
    \begin{minipage}{.4\textwidth}
        \centering
        \includegraphics[width=1.0\linewidth, trim = 20pt 20pt 20pt 20pt]{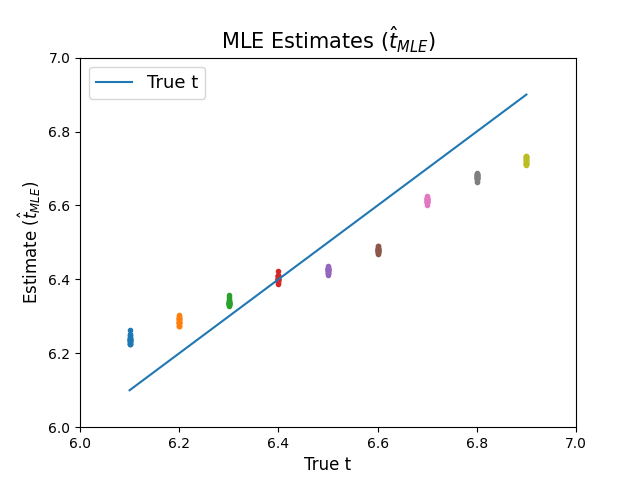}
    \end{minipage}
    \begin{minipage}{.4\textwidth}
        \centering
        \includegraphics[width=1.0\linewidth, trim = 20pt 20pt 20pt 20pt]{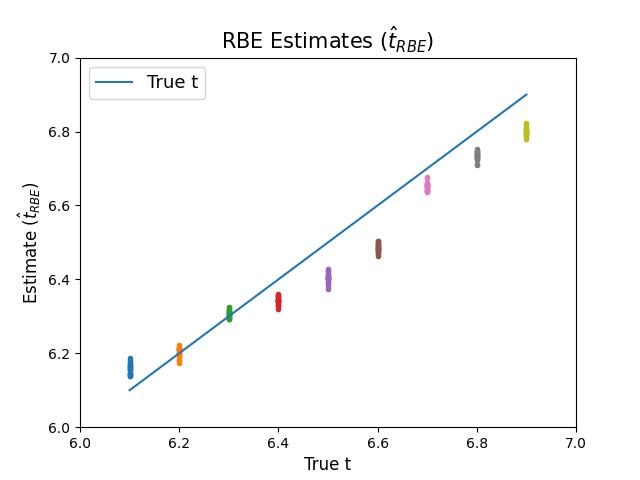}
    \end{minipage}

    \bigskip
    \begin{minipage}{.4\textwidth}
        \centering
        \includegraphics[width=1.0\linewidth, trim = 20pt 20pt 20pt 20pt]{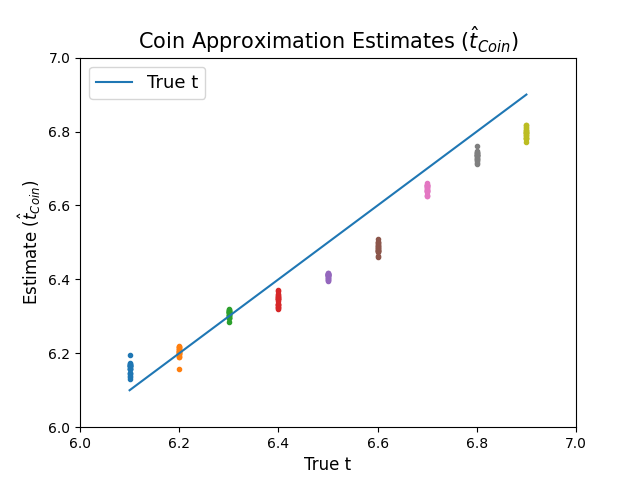}
    \end{minipage}
    \begin{minipage}{.4\textwidth}
        \centering
        \includegraphics[width=1.0\linewidth, trim = 20pt 20pt 20pt 20pt]{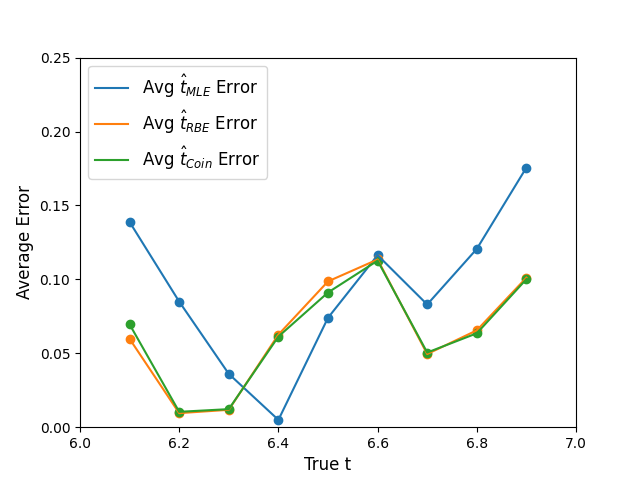}
    \end{minipage}

    \captionof{figure}{Top Left: MLE estimates for $t$ computed using 1,000 samples from experiments on
    \textit{ibm\_perth} with $n = 3$ qubits and values of $t$ at increments of 0.1 in between 6.1
    and 6.9. Top Right: RBE estimates for $t$ computed using 1,000 samples from experiments on
    \textit{ibm\_perth} with $n = 3$ qubits and values of $t$ at increments of 0.1 in between 6.1
    and 6.9.. Bottom Left: Coin Approximation estimates for $t$ computed using 1,000 samples from experiments on
    \textit{ibm\_perth} with $n = 3$ qubits and values of $t$ at increments of 0.1 in between 6.1 and 6.9. Bottom
    Right: The average error of estimates computed using samples from experiments on a quantum simulator.}
    \label{fig:perth_estimations_t6}
\end{figure*}

The experiment was repeated for various intervals of $t$ using 3 and 4 qubit registers on \textit{ibm\_perth}.
Results from these additional experiments are included in Appendix~\ref{sec:add_experiments}.

\section{\label{sec:conclusions}Concluding Remarks}

In this paper we explore a few methods for estimating the decimal part of a number encoded through the Phase or
Amplitude Estimation Algorithm.

The full interpolation method and the Maximum Likelihood Estimation method use sampling counts for all possible
outcomes, and give good theoretical results, but it turns out that they are sensitive to noise when implemented
on quantum hardware that is currently available, without additional error correction.
Methods that rely on only the top two counts seem to be less sensitive to such noise, and they are also simple to use,
providing closed forms for the decimal of the encoded value for the given counts.
Confidence/credible intervals for these methods are also relatively simple to obtain.

\acknowledgements

The views expressed in this article are those of the authors and do not represent the views of Wells Fargo. This
article is for informational purposes only. Nothing contained in this article should be construed as investment advice.
Wells Fargo makes no express or implied warranties and expressly disclaims all legal, tax, and accounting implications
related to this article.\\

We acknowledge the use of IBM Quantum services for this work. The views expressed are those of the authors, and do
not reflect the official policy or position of IBM or the IBM Quantum team.

\bibliographystyle{unsrtnat}
\bibliography{main}

\renewcommand\floatpagefraction{0.9}
\appendix

\section{\label{sec:other_properties}Addtional Properties of the Discrete Sinc Quantum State}

The following additional identities can be useful in further understanding the discrete sinc quantum state and associated
Fej\'er distribution.

\begin{lemma}
    Given a non-integer $t \in (0, N)$ and an integer $0 \le k < N$ we have

    \begin{equation*}
        \begin{split}
            \tan((t-k)\frac{\pi}{N}) &= \sqrt{\frac{p_{N, t}\left((k + \frac{N}{2}) \mod N\right)}{p_{N, t}(k)}} \\
        \end{split}
    \end{equation*}

    \begin{equation*}
        \sum_{k = 0}^{N-1} \sin((t-k)\frac{\pi}{N}) c_{N, t}(k) = 0
    \end{equation*}

    \begin{equation*}
        \sum_{k = 0}^{N-1} \cos((t-k)\frac{\pi}{N}) c_{N, t}(k) = 1
    \end{equation*}

    \begin{equation*}
        \sum_{k=0}^{N-1} (-1)^k \sin((t-k)\frac{\pi}{N}) c_{N, t}(k) = \sin(t\pi)
    \end{equation*}

    \begin{equation*}
        \sum_{k = 0}^{N-1} e^{i \pi \frac{N-1}{N} k} c_{N,t}(k)  =  e^{i \pi \frac{N-1}{N} t}
    \end{equation*}

    \begin{equation*}
        \sum_{k=0}^{N-1} (-1)^k \cos(k\frac{\pi}{N}) c_{N, t}(k) = \cos(t\frac{N-1}{N}\pi)
    \end{equation*}

    \begin{equation*}
        \sum_{k=0}^{N-1} (-1)^k \sin(k\frac{\pi}{N}) c_{N, t}(k) = -\sin(t\frac{N-1}{N}\pi)
    \end{equation*}

    \begin{equation*}
        \frac{1}{N}\sum_{k=0}^{N-1} \cot((t-k)\frac{\pi}{N}) = \cot(t\pi)
    \end{equation*}
\end{lemma}

\section{\label{sec:stats}Estimator Analysis}

For a positive integer $L$, and a sequence of $L$ measurements of the state $\ket{\phi_{n, t}}$ defined in Eq
.~\ref{eqn:fejer_state} (where $n$ is a positive integer, $N = 2^n$ and $t \in [0, N)$), we denote by $r$ the ratio
of the probabilities of the states $\ket{k}$ and$\ket{k+1}$.
We denote by $\hat{r}$ the estimator of this ratio.
Following ~\cite{meanandvariance2018}, we analyze the expectation and variance of $\hat{r}$.

The expectation of $\hat{r}$ is

\begin{equation}
    \label{eqn:ratio_mean}
    \mu_{\hat{r}} = \frac{p_{N, t}(k)}{p_{N, t}(k+1)} \left( 1 + \frac{1}{L p_{N, t}(k+1)} \right),
\end{equation}

and its variance is

\begin{equation}
    \label{eqn:ratio_variance}
    \sigma_{\hat{r}}^2 = \frac{1}{L} \left(\frac{p_{N, t}(k)}{p_{N, t}(k+1)}\right)^2\left(\frac{1}{p_{N, t}(k)} +
    \frac{1}{p_{N, t}(k+1)}\right)
\end{equation}

where $p_{N, t}$ is defined in Eq.~\ref{eqn:fejer_pdf} and $0 \le k < N$.

Figure~\ref{fig:ratio_plot} shows samples of the expectation and variance compared to the theoretical values.

\begin{center}
    \begin{tabular}{cccc}
        \begin{minipage}{.4\textwidth}
            \centering
            \includegraphics[width=.9\linewidth]{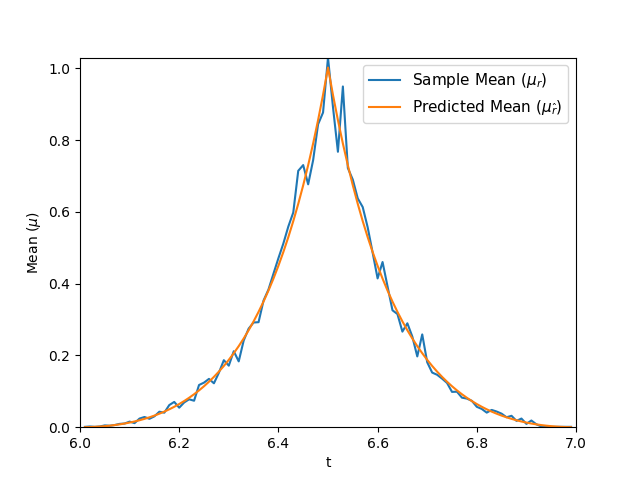}
        \end{minipage}
        \begin{minipage}{.4\textwidth}
            \centering
            \includegraphics[width=.9\linewidth]{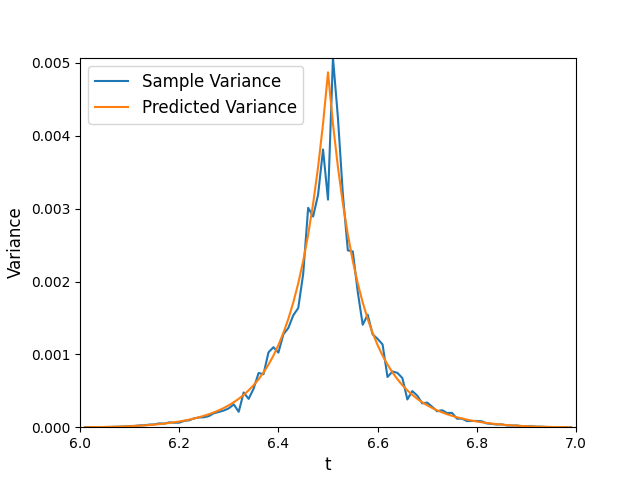}
        \end{minipage}
    \end{tabular}
\captionof{figure}{Left: The expectation of $\hat{r}$ defined in Eq.~\ref{eqn:ratio_mean} compared to measured values
of $r$ from a quantum simulator for $L = 1,000$ samples, $n = 3$ qubits and values of $t \in (6, 7)$. Right: The
predicted variance of $\hat{r}$ defined in Eq.~\ref{eqn:ratio_variance} compared to the variance of measured values
of $r$ from 100 rounds using a quantum simulator for $L = 1,000$ samples, $n = 3$ qubits and values of $t \in (6, 7)$.
In both figures, we use the opposite ratio for values of $t$ with a decimal part greater than 0.5 to keep the ratio
less than 1.}
    \label{fig:ratio_plot}
\end{center}

Using the Taylor expansion of the function $D_N$ defined in Eq.~\ref{eqn:d_definition}, we can compute
the expectation and variance of the estimator.

The expectation is

\begin{equation}
    \label{eqn:rbe_mean}
    E (\hat{t}_{\text{RBE}}) \approx D_N^{\prime}(\mu_{\hat{r}}) + \frac{D_N^{\prime\prime}(\mu_{\hat{r}})}{2}
    \sigma_{\hat{r}}^2,
\end{equation}

and the variance is

\begin{equation}
    \label{eqn:rbe_variance}
    \text{Var}(\hat{t}_{\text{RBE}}) \approx  D_N^{\prime}(\mu_{\hat{r}})^2 \sigma_{\hat{r}}^2 - \frac{1}{4}D_N^{\prime\prime}
    (\mu_{\hat{r}})^2 \sigma_{\hat{r}}^2
\end{equation}

where $\hat{t}_{\text{RBE}}$ is defined in Eq.~\ref{eqn:ratio_est}.

\begin{center}
    \begin{tabular}{cccc}
        \begin{minipage}{.4\textwidth}
            \centering
            \includegraphics[width=.9\linewidth]{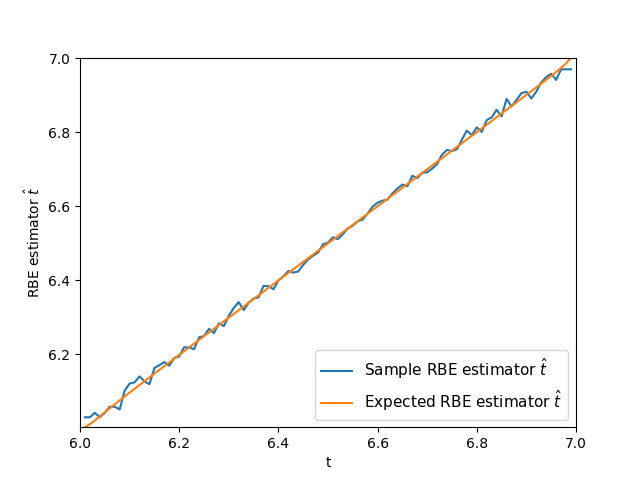}
        \end{minipage}
        \begin{minipage}{.4\textwidth}
            \centering
            \includegraphics[width=.9\linewidth]{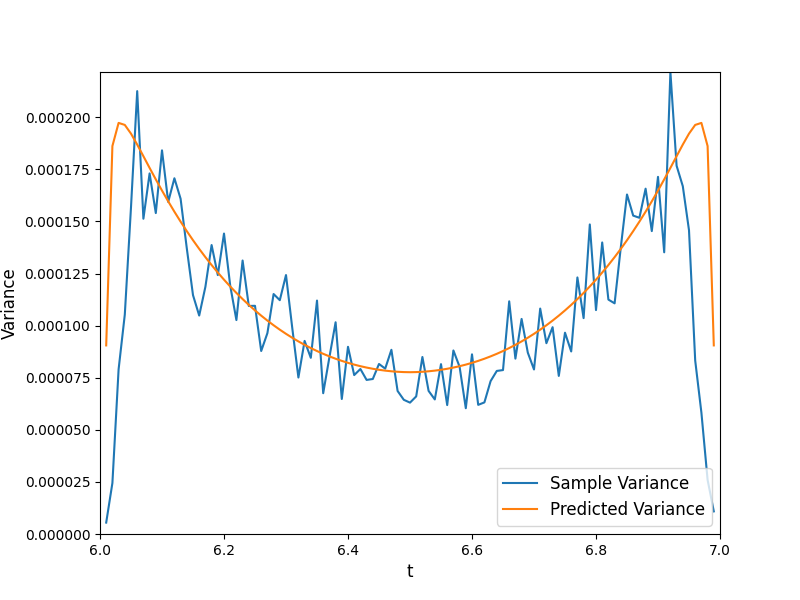}
        \end{minipage}
    \end{tabular}
    \captionof{figure}{Left: The expectation of $\hat{t}_{\text{RBE}}$ defined in Eq.~\ref{eqn:rbe_mean}, with
        $\mu_{\hat{r}}$ defined in Eq.~\ref{eqn:ratio_mean}, $\sigma_{\hat{r}}^2$ defined in Eq
        .~\ref{eqn:ratio_variance} and $L = 1000$, compared to $\hat{t}_{\text{RBE}}$ computed using 1,000 samples
        from a quantum simulator with $n = 3$, and values of $t \in (6, 7)$. Right: The variance of
        $\hat{t}_{\text{RBE}}$ defined in Eq.~\ref{eqn:rbe_variance}, with $\mu_{\hat{r}}$
         defined in Eq.~\ref{eqn:ratio_mean}, $\sigma_{\hat{r}}^2$ defined in Eq.~\ref{eqn:ratio_variance} and
        $L = 1000$, compared to the variance of 100 $\hat{t}_{\text{RBE}}$ values computed using 1,000 samples from a
        quantum simulator.}
    \label{fig:re_plot}
\end{center}

\section{\label{sec:hardware}Hardware Details}

These experiments were run on the IBM Quantum system \textit{ibm\_perth}, which a IBM Quantum Falcon processor with 7
qubits.
The error map at the time of the experiments is shown in Figure~\ref{fig:perth_map}. Each circuit in the
experiment was measured with $20,000$ shots.
Relevant calibration data is included in Figure~\ref{fig:perth_map} and Table~\ref{tab:perth_table}.

\begin{figure}
    \centering
    \begin{minipage}{.35\textwidth}
        \centering
        \includegraphics[width=.95\linewidth]{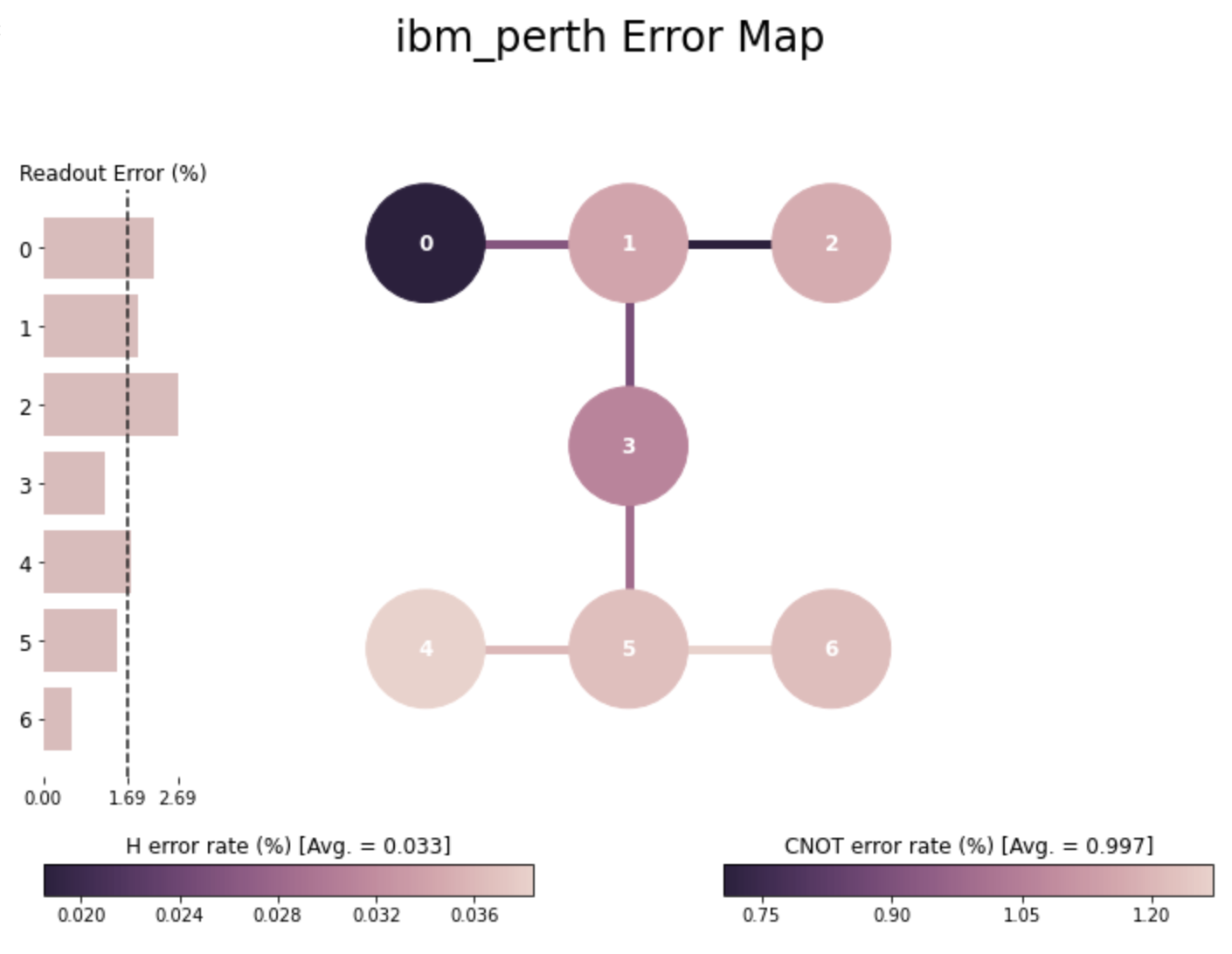}
    \end{minipage}
    \caption{Error map of \textit{ibm\_perth} generated at the time of the experiments. The three qubit circuits were
    run on qubits ${1, 3, 5}$. The four qubit circuits were run on ${3, 4, 5, 6}$. The six qubit circuits were run
    on ${0, 1, 3, 4, 5, 6}$.}
    \label{fig:perth_map}
\end{figure}

\begin{table}
    \centering\small
    \begin{tabular}{c c c | c c c} \hline\hline
            Qubit pair & $\quad$Error (\%) & & Qubit & $\quad T_1~(\mu{\rm s})$ & $\quad T_2~(\mu{\rm s})$ \\ \hline
            (0, 1) & 0.918 & & Q0 & 147 & 85\\
            (1, 3) & 0.890 & & Q1 & 218 & 57  \\
            (3, 5) & 0.990 & & Q3 & 130 & 127 \\
            (4, 5) & 1.203 & & Q4 & 155 & 165 \\
            (5, 6) & 1.271 & & Q6 & 199 & 155 \\
            Average & 1.05$\pm$0.17 & &  &  \\ \hline
    \end{tabular}
    \caption{Calibration data from \textit{ibm\_perth} from calibration before experiments.}
    \label{tab:perth_table}
\end{table}

\newpage
\section{\label{sec:add_experiments}Additional Quantum Hardware Experiments}

Figures~\ref{fig:perth_n3} and ~\ref{fig:perth_n4} visualize the results of experiments run on \textit{ibm\_perth}
with 3 and 4 qubits, respectively, for more intervals of $t$.

\begin{figure*}[h!]
    \centering
    \begin{minipage}{.4\textwidth}
        \centering
        \includegraphics[width=.85\linewidth, trim = 20pt 20pt 20pt 20pt]{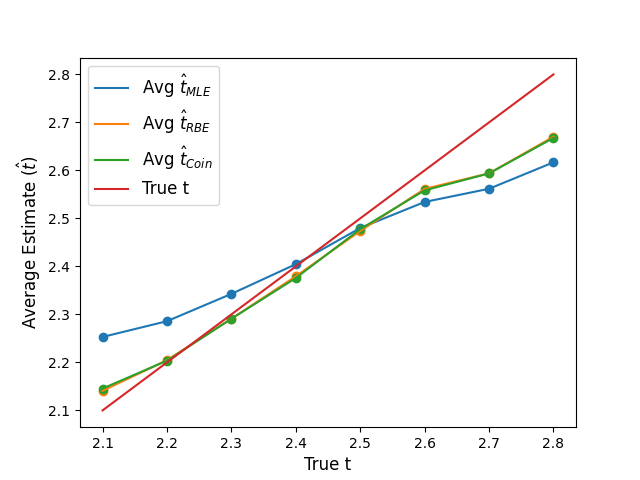}
    \end{minipage}
    \begin{minipage}{.4\textwidth}
        \centering
        \includegraphics[width=.85\linewidth, trim = 20pt 20pt 20pt 20pt]{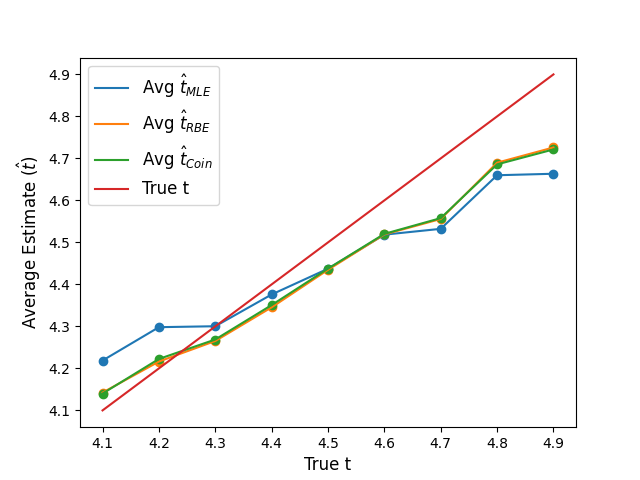}
    \end{minipage}

    \bigskip
    \begin{minipage}{.4\textwidth}
        \centering
        \includegraphics[width=.85\linewidth, trim = 20pt 20pt 20pt 20pt]{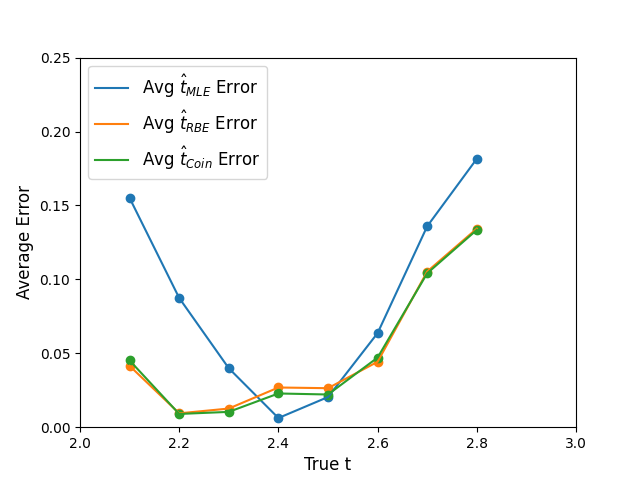}
    \end{minipage}
    \begin{minipage}{.4\textwidth}
        \centering
        \includegraphics[width=.85\linewidth, trim = 20pt 20pt 20pt 20pt]{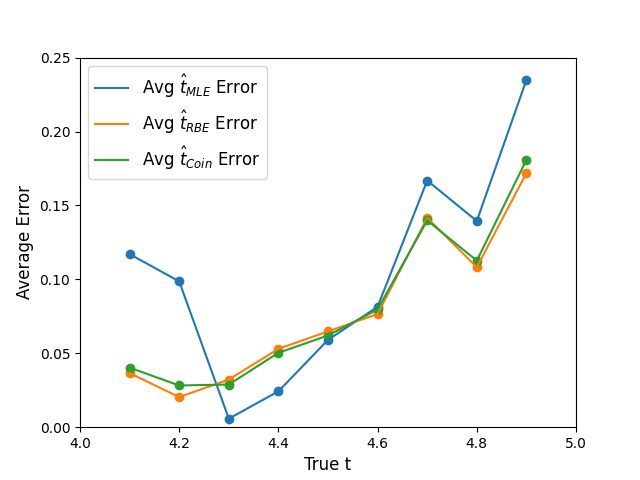}
    \end{minipage}

    \captionof{figure}{Left: The average of 20 estimates using each of the three methods (MLE, RBE, and Coin
    Approximation) from experiments on \textit{ibm\_perth} with $n = 3$ qubits and values of $t$ at increments of 0
    .1 in between 2.1 and 2.9 and the average error of the respective estimates. Right: The average
    of 20 estimates using each of the three methods (MLE, RBE, and Coin Approximation) from experiments on
    \textit{ibm\_perth} with $n = 3$ qubits and values of $t$ at increments of 0.1 in between 4.1 and 4.9, and
    the average error of the respective estimates.}
    \label{fig:perth_n3}
\end{figure*}

\begin{figure*}[p]
    \vspace{-6in}
    \centering
    \begin{minipage}{.4\textwidth}
        \centering
        \includegraphics[width=.85\linewidth, trim = 20pt 20pt 20pt 20pt]{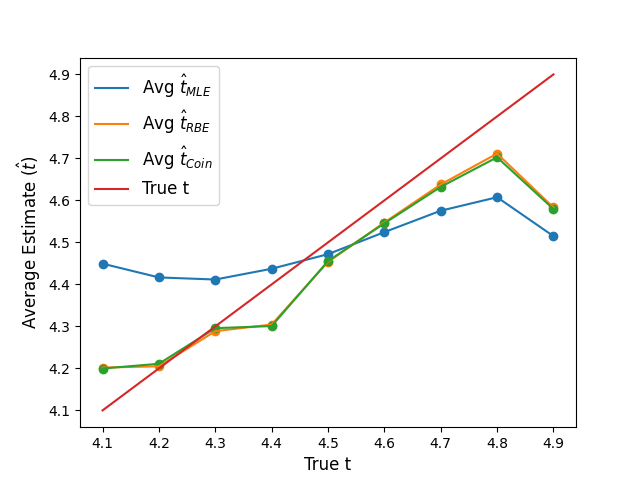}
    \end{minipage}
    \begin{minipage}{.4\textwidth}
        \centering
        \includegraphics[width=.85\linewidth, trim = 20pt 20pt 20pt 20pt]{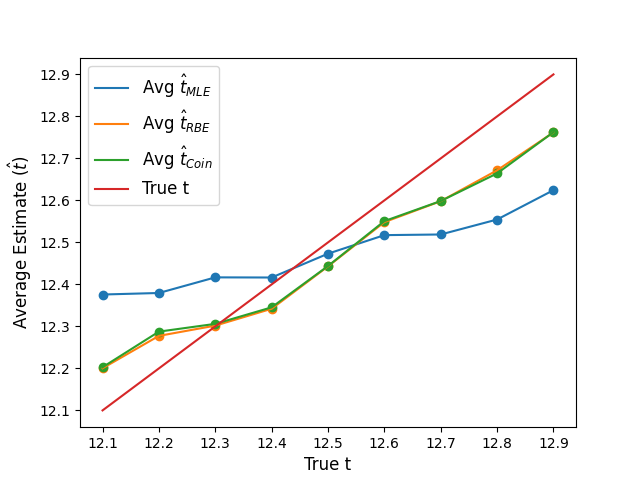}
    \end{minipage}

    \bigskip
    \begin{minipage}{.4\textwidth}
        \centering
        \includegraphics[width=.85\linewidth, trim = 20pt 20pt 20pt 20pt]{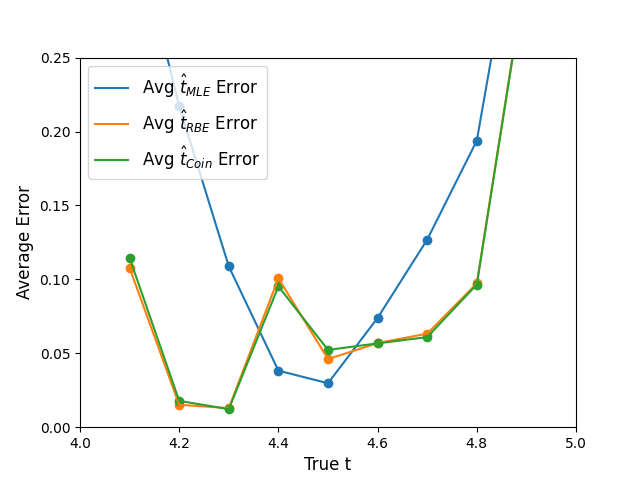}
    \end{minipage}
    \begin{minipage}{.4\textwidth}
        \centering
        \includegraphics[width=.85\linewidth, trim = 20pt 20pt 20pt 20pt]{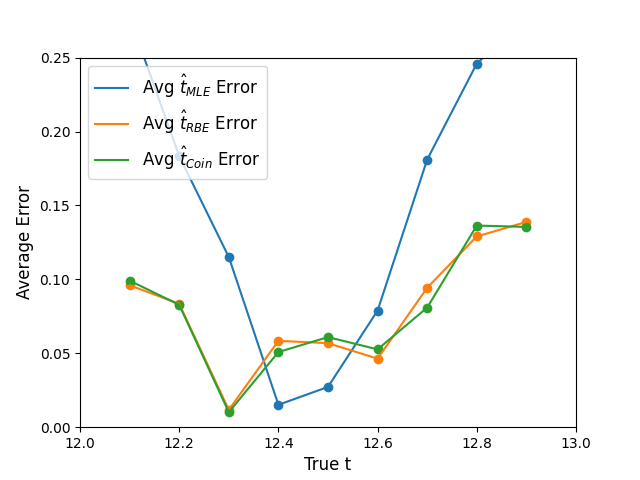}
    \end{minipage}

    \captionof{figure}{Left: The average of 20 estimates using each of the three methods (MLE, RBE, and Coin
    Approximation) from experiments on \textit{ibm\_perth} with $n = 4$ qubits and values of $t$ at increments of 0
    .1 in between 4.1 and 4.9, and the average error of the respective estimates. Left: The average of 20 estimates
    using each of the three methods (MLE, RBE, and Coin Approximation) from experiments on \textit{ibm\_perth} with
        $n = 4$ qubits and values of $t$ at increments of 0.1 in between 12.1 and 12.9, and the average error of the
        respective estimates.}
    \label{fig:perth_n4}
\end{figure*}

\end{document}